\documentclass[]{article}

\usepackage{amsmath} \usepackage{amsfonts}
\usepackage{amssymb} \usepackage{amsthm} 

\usepackage{enumerate} 
\usepackage{algorithm}
\usepackage{algpseudocode} 
\usepackage{authblk} 
\usepackage{bbding} 

\newcommand{\bfa}[1]{\mbox{\boldmath $ #1 $}}

\usepackage{graphicx}

\title{A Bayesian approach to parameter identification with
an application to Turing systems}
\author[1,\Envelope]{Eduard Campillo-Funollet}
\author[2]{Chandrasekhar Venkataraman}
\author[3,\Envelope]{Anotida Madzvamuse}
\affil[1]{\small University of
Sussex, School of Mathematical and Physical Sciences, Brighton, UK. {\it e.campillo-funollet@sussex.ac.uk}\newline}
\affil[2]{University of St Andrews, School of Mathematics and Statistics, UK. {\it cv28@st-andrews.ac.uk}\newline}
\affil[3]{University of
Sussex, School of Mathematical and Physical Sciences, Brighton, UK. {\it a.madzvamuse@sussex.ac.uk}\newline}
\affil[\Envelope]{\small{Corresponding author}}

\renewcommand{\d}{\mathrm{d}}

\newtheorem{theorem}{Theorem}[section]
\newtheorem{proposition}{Proposition}[section]

\begin{document}
\maketitle 

\begin{abstract}
 We present a Bayesian methodology for infinite as well as finite dimensional parameter identification for
partial differential equation models. The Bayesian framework provides a rigorous mathematical framework for incorporating prior knowledge on uncertainty in the observations and the parameters themselves, resulting in an
approximation of the full probability distribution for the parameters, given the
data. Although the numerical approximation of the full probability distribution
is computationally expensive, parallelised algorithms can make many practically
relevant problems computationally feasible. The probability distribution not
only provides estimates for the values of the parameters, but also provides
information about the inferability of parameters and the sensitivity of the
model. This information is crucial when a mathematical model is used to study
the outcome of real-world experiments. 

Keeping in mind the applicability of our approach to tackle  real-world practical problems with data from  experiments, in this initial proof of concept work, we apply this theoretical and computational framework to parameter identification for a
well studied semilinear reaction-diffusion system with {\it activator-depleted}
 reaction kinetics, posed on evolving and stationary domains.
\end{abstract}


\section{Introduction}
\subsection{Motivation and highlights}

The Bayesian approach for infinite dimensional parameter identification problems, that we elucidate in the present work, 
is a powerful tool that allows mathematically, rigorous qualitative and quantitative inference to be made on the parameters in a mathematical model from the results of experimental studies of the system that is being modelled. Some aspects of the approach that we describe in detail of the remainder of this work are as follows;
\begin{itemize}
\item One seeks to find the conditional probability distribution of the  parameters given the experimental data. For many problems in the life sciences, variation in parameters, across organisms
for example, is commonplace. In such a setting, more widespread approaches that seek to identify a single parameter value, i.e., the most likely value given the data, may be next to worthless, whilst identifying the conditional probability distribution is significantly more informative. One example that we illustrate below is the computation of credible regions, the analogue of confidence intervals, for parameters. Such regions arise naturally from the approach and are typically tighter than those obtained using a method which seeks only a single value.
\item Uncertainty associated with the measurements from experiments may be naturally incorporated in the framework. As experimental techniques and the associated confidence in them and understanding of them matures, such 
a framework becomes particularly attractive.
\item Recent advances in the development of fast and efficient parallel algorithms can be employed in order to reduce the large computation times associated with the approach. We illustrate this by describing a robust and efficient parallel algorithm for the sampling of the aforementioned probability distribution in the sequel.
\item It is possible to show the problem of parameter identification is well posed (in a mathematical sense) under very mild assumptions on the model and the experimental data. For example, in the present work we consider parameter identification for a pattern forming system of equations. Proving the problem is well posed requires little more effort than showing the pattern forming model itself is well posed. 
\end{itemize}
The upshot is that we feel such a framework is ideally suited to tackling challenging parameter identification problems in biology, such as those in cell biology,  where the natural variability in experimental measurements together with the prevalence of highly complex nonlinear mathematical models motivates an approach with the features described above.

\subsection{Overview on parameter identification}
In the study of a real world system, experiments provide measurements of
quantities from the system, whilst theoretical work provides models predicting
the behaviour of the real world system. Such models typically include
 independent variables---the parameters---inherent to the system under
study. 

Parameter identification is the problem of extracting information about the
parameters of a model from the actual result of some measurement 
\cite{tarantola05}. We consider mathematical models of the form of partial differential equations (PDEs) whose
variables are elements of a Banach space.  The
data will be a noisy measurement of the solution of the system of PDEs, whilst the parameters will be the coefficients
of the mathematical  model, i.e. the coefficients of the system of PDEs,
including maybe initial and boundary conditions which can be considered as
parameters as well. The data may also be a functional of the solution, such as an average. Parameter
identification problems are sometimes referred to as inverse problems, thus
considering as a direct problem solving the system of PDEs given the
parameters. Note that parameters may be functions in general. 
We remark that this approach applies directly to systems of ordinary
differential equations as well, making the theory of interest to a wide range
of fields. 
 
There are two main approaches to solve parameter identification problems. On
one hand, one can use techniques from optimal control to find the best possible
value for the parameter. Here, the best possible value is defined in terms of
the distance from the solution of the model to the given data. In a Banach
space setting, when the data and the solution are elements of a Banach space,
this corresponds to the problem of minimising the norm of the difference
between the solution and the data. In general, this problem is ill-posed. A
regularisation term must be added, enforcing the norm
of the parameter to be minimised too, and  effectively restricting the candidate
parameters. Ill-posedness may come from different sources, like non-uniqueness
or non-continuous dependence of the solution with respect to the data (see, for
instance, \cite{aster13,beck85}).

A second approach to parameter identification is to use Bayesian techniques
(see for instance \cite{stuart10,kaipio06}).
This approach consists of a set of techniques built around Bayes' theorem, which roughly
speaking allows one to compute the conditional probability of an event $A$ given
the probability of an event $B$, in terms of the reverse conditional probability, i.e. the
probability of $B$ given $A$. In the parameter identification framework, this
will correspond to computing the probability of the parameters given the data, in
terms of the probability of the data given the parameters. Note that the latter
is given by the mathematical model: when the parameters are given, the
mathematical model predicts the result of the experiment, i.e. the
data.

In both cases, numerical algorithms are necessary to approximate the final
result. In the optimal control framework, there is a range of minimisation
techniques that may be used, for instance Gauss-Newton optimisation or the
Levenberg-Marquardt algorithm (see \cite{adby13} for a general description of
these algorithms). These algorithms provide the best value for the
parameters, in the sense of minimising the distance from the data to the
solution of the mathematical model. Similarly, an optimal value can be obtained
from the Bayesian framework using maximum likelihood methods
\cite{muller04,timmer04,baker05}. Additional methods can be used to estimate
confidence intervals for the parameters. In the Bayesian framework, the problem will be to
estimate a full probability distribution. In this case, usually a Markov Chain
Monte Carlo algorithm is used to build a Markov chain with stationary
distribution equal to the target distribution (see
\cite{norris98,stuart10,toni09,battogtokh02,brown03}). This usually requires a lot of
evaluations of the observation functional---the mapping from parameters to
observations, i.e. from parameters to the solution of the mathematical model---, solving the system of PDEs for each
evaluation. 
Whilst this is computationally more expensive than a minimisation algorithm,
the amount of information obtained---the full probability distribution---is
larger compared to the single value that we get from the optimal control
approach. Note that in both cases a parallel algorithm can be used to speed-up the
computations.

In contrast with the value for the parameters given by the optimal control
approach, the Bayesian framework provides a probability distribution. From the
probability distribution one can extract information about the uncertainty of
the parameters, thus avoiding situations where despite obtaining a value for
the parameters, a more detailed study of the uncertainty shows that they are
not well-determined \cite{gutenkunst07}. In many practical situations, the exact value of the
parameters is not useful in drawing conclusions. For instance, in
\cite{ashyraliyev08} it is shown that correlations between parameters in a
model for \textit{Drosophila Melanogaster} gap gene impede the individual
identification of the parameters, but it is still possible to obtain
conclusions about the regulatory topology of the gene network. Furthermore, the
information provided by the probability can be used in model selection problems
\cite{vyshemirsky08}, where several mathematical models are available and the
goal is to find which one fits better to the experimental data.

There are other advantages inherent to the fact that a probability distribution
for the parameters is available. For example, correlations between different
parameters can be observed in the joint probability distribution. These
correlations can be invisible when only a value for the parameters is computed,
leading to an overestimation of the uncertainty \cite{sutton16}. Similarly, the
correlations between parameters can suggest a non-dimensionalised form of the
equations to eliminate non-relevant parameters from the problem. 

A drawback of the Bayesian approach is its computational cost. To obtain a good
approximation of the posterior probability distribution, we need to solve the
system of partial differential equations many times. The use of
parallelised algorithms makes many relevant problems feasible. The same
problems were not suitable for this approach a few years ago due to the lack of
computational power \cite{siam01}.

Parameter identification problems often arise in many areas of research such as biology, biomedicine, chemistry, material sciences, mathematics, statistics etc.
Typically, a mathematical model needs to be fitted with experimental data.
Only in a few cases is it possible to do a
analytical parameter identification, see for instance \cite{friedman92} . Note also that parameters are not always
directly measurable from experiments, or have a defined physical meaning.
Hence, parameter identification through the proposed approach might help guide
experimentalists in identifying regions of interest to which parameters belong.

The optimal control approach is widely used. In \cite{ashyraliyev09}, it is
applied to systems of ordinary differential equations modelling enzyme kinetics. In this
case, confidence regions may be computed using the second derivatives of the
cost functional that is being minimised. Since this only provides local
information around the optimal point, the confidence regions computed with this
method tend to be large. The optimal control approach has also been applied to
geometric evolution problems for cell motility, and to Turing systems
\cite{garvie10,croft15,blazakis15,portet15,yang15}.

Bayesian techniques are common for parameter fitting in statistics. The main
difference with the present work is the type of models. We use models given by
partial differential equations instead of statistical models. In particular,
our models are deterministic---with a uniquely determined solution for each
admissible parameter.  Statistical models, or stochastic models in general,
incorporate randomness in the model. The Bayesian approach to parameter
identification for models given by PDEs (see \cite{stuart10} and the references
therein) is used in different fields, for example inverse problems in
subsurface flow \cite{iglesias14}.

In the absence of a real experimental model, for illustrative purposes, we consider a well known system of partial differential equations of the form of reaction-diffusion type, also known as Turing systems \cite{murray11,murray13,turing52}. The motivation is that parameter space identification for some reaction-diffusion system with classical reaction kinetics can be obtained analytically thereby offering us a bench-marking example for our theoretical and computational framework. Turing systems are a family of reaction-diffusion systems introduced by Turing in \cite{turing52}.
The main feature of these models is that small perturbations of homogeneous
steady states may evolve to solutions with non-trivial patterns
\cite{murray11,murray13}. An example
of such  a model is the Schnakenberg system (also known as the {\it
activator-depleted} substrate model) \cite{gierer72,prigogine68,schnakenberg79}. This model has
been widely studied, both analytically and numerically \cite{garvie10,madzvamuse05}. 
Our first example is to identify a time-dependent growth rate function of the reaction-diffusion system on a one-dimensional growing domain. For this case, the rest of the model parameters are considered known and fixed. This example illustrates the infinite dimensional parameter identification example, rather than parameters. Our second example is that of parameter identification in a finite dimensional framework  we consider the reaction-diffusion system posed on a stationary two-dimensional domain. For both examples, we use synthetic patterns  -- computer generated -- based on previous works using a fixed set of
parameters.

\section{Methodology on parameter identification}
The main ingredients of a parameter identification problem are the data (the
measurement from the experiment), the parameters and the mathematical model.
The data and the parameters are modelled as elements of suitable Banach spaces.
To fix the notation, let $U$ and $Y$ be Banach spaces, let $y\in Y$ be the data,
and let $p \in U$ be the parameter. Note that in general $p$ is a vector of
parameters. When one fixes a value for the parameter $p$, one can solve
the equations of the model to find the solution associated with this parameter.
This process is known as the solution of the forward problem; in contrast, the
problem of parameter identification is known as the inverse problem. The
solution of the forward problem is formalised with a mapping $G:U\rightarrow
Y$, such that a parameter $p\in U$ is mapped to the solution of the model with
parameter $p$. Since $G$ maps
parameters to the measurement space, $G$ is called the observation operator.
$G$ is well-defined as long as for each possible parameter, there exists a
unique solution to the mathematical model, otherwise $G$ might not be defined
in $U$ or might be multivalued.  

The measurements or observations generated by experiments typically have some 
error or uncertainty associated with them.
Following the usual terminology, we refer to such errors as the noise of the
measurement. We consider here additive noise, i.e. an error added to the true
value of the measurement. Note that different sources of error can be taken
into account and these can be included as several independent additive errors or as extra parameters in the
mathematical model; different types of noise, e.g. multiplicative, can also be
incorporated in this way \cite{glimm03,orrell99,osullivan06,cotter09}.

The noise is naturally modelled in terms of a probability distribution. Let
$\mathbb{Q}_0$ be a probability distribution. 
Let $\eta$ be a realisation of $\mathbb{Q}_0$. With this notation at hand, our
measurement can be expressed as
\begin{equation}\label{eqn:meas_model}
y = G(p) + \eta.
\end{equation}

The Bayesian approach to parameter identification assumes implicitly that a
probability distribution models the state of knowledge about a quantity \cite{tarantola05}. If the quantity is known with infinite precision then the
probability distribution will be concentrated at one point---a Dirac delta
distribution. The less we know, the more spread the distribution will be.

Equation \eqref{eqn:meas_model} expresses the relation between the probability
distribution of the data $y$, the parameters $p$ and the noise $\eta$. In other
words, \eqref{eqn:meas_model} connects our knowledge about the different
ingredients of the parameter identification problem. Note that the mathematical
model, encoded in $G$, is assumed to be known; our interest is to find
information about the parameters of a given model, not to find the model
itself, although dummy parameters could be used for model selection.

From the modelling point of view, equation \eqref{eqn:meas_model} is
well-defined as long as the data $y$ lies in a Banach space. Some systems may
require a more general framework, for instance if the data lies in a manifold
or in a metric space, where the operation "+" may not be defined. The Banach
space setting allows us to prove rigorous results on the parameter identification
problem, and it is suitable for the numerical approximation of the solution
\cite{stuart10,dashti13,cotter13}. In terms of applicability, in some
situations the task of finding the noise distribution is not trivial
\cite{kaipio07,huttunen07}. 

The parameter identification problem can be now stated as follows: find the
probability distribution of the parameter $p$ given the data $y$. This
probability distribution is called the posterior, and is denoted by
$\mathbb{P}(p|y)$, or $\mathbb{P}^y(p)$.
All the information available about the parameter $p$ is encoded in the
posterior probability distribution. The goal now is to characterise this
distribution.

From \eqref{eqn:meas_model} we can compute directly the probability
distribution of the data $y$ given a parameter $p$. Since we know the noise
distribution $\mathbb{Q}_0$, given a parameter $p$, the distribution of 
$y$ is $\mathbb{Q}_0$ shifted by $G(p)$. In terms of the probability density
function (PDF) of $\mathbb{Q}_0$, denoted by $\rho$, the PDF of the distribution
of the data $y$, given a parameter $p$, is $\rho(y-G(p))$. In the
sequel it will be useful to use the so called log-likelihood of the parameter
$p$, given by \cite{stuart10,dashti13} 
\begin{equation}\label{eqn:logl}
\Phi(p;y) = -\log(\rho(y-G(p))).
\end{equation}

Bayes' theorem characterises a conditional probability distribution
$\mathbb{P}(p|y)$ in terms of the reverse conditional probability distribution
$\mathbb{P}(y|p)$ and the marginal distribution $\mathbb{P}(p)$. The latter can
be interpreted as the probability distribution of $p$, i.e. the state of
knowledge about the parameter, regardless of the data $y$. Bayes' theorem
gives
\begin{equation}\label{eqn:bayes}
\mathbb{P}(p|y) \propto \mathbb{P}(y|p) \mathbb{P}(p).
\end{equation}
The constant of proportionality can be computed explicitly but we do not need
it, because the algorithm that we will use to numerically approximate the
probability distribution only uses the ratio between probabilities; the
constant of proportionality is itself cancelled.
 
The probability distribution $\mathbb{P}(p)$ is called the prior distribution,
and encodes our prior knowledge about the parameter $p$. This knowledge comes from
any source that was available before the performance of the experiment that
provides the data for the parameter identification. If there is not a previous
knowledge about the parameter $p$, the absence of knowledge can also be
represented. These type of priors are sometimes called non-informative priors
\cite{kass96}. In a finite dimensional setting, a natural choice for a
non-informative prior is the uniform distribution.

We can now apply Bayes' theorem to \eqref{eqn:meas_model}. Let $\mu^y$ be
the probability density function of $\mathbb{P}(p|y)$. We have seen that the
PDF of $\mathbb{P}(y|p)$ is given by $\rho(y-G(p))$, where $\rho$ is the PDF of
the noise distribution $\mathbb{Q}_0$. Denote by $\mu_0$ the PDF of the prior
distribution. From Bayes' theorem we obtain \cite{stuart10}
\[
\mu^y(p) \propto \rho(y-G(p))\mu_0(p).
\]

Bayes' theorem can be generalised to an infinite dimensional setting. 
The formulation is more technical but analogous. With
the general version of the theorem, the Bayesian approach to parameter
identification can be applied to infinite dimensional problems. Let
$\mathbb{P}_0$ be the prior distribution. The following result generalises the
Bayes' theorem to infinite dimensional settings \cite{stuart10,dashti13}.

\begin{theorem}[Bayes' theorem]\label{prop:bayes}
Assume that the log-likelihood $\Phi: U \times Y \longrightarrow \mathbb{R}$ is
measurable with respect to the product measure $\mathbb{P}_0 \times
\mathbb{Q}_0$ and assume that 
\[
Z := \int_U \exp\left( -\Phi(p;y) \right) \mathbb{P}_0(\d p) > 0.
\]
Then, the conditional distribution $\mathbb{P}(p|y)$, denoted by
$\mathbb{P}^y(p)$, exists, and it is absolutely continuous with respect to the
prior distribution $\mathbb{P}_0$. Furthermore,
\[
\frac{\d \mathbb{P}^y}{\d \mathbb{P}_0} = \frac{1}{Z} \exp\left( - \Phi(p;y)
\right).
\]
\end{theorem}
\begin{proof}  See \cite[theorem 3.4]{dashti13}. 
\end{proof}

Using only properties of the forward problem, the following proposition (see
\cite{stuart10}) guarantees that the parameter identification problem is
well-posed: the posterior distribution exists, and it depends continuously on
the data $y$. Denote by $B(0,r)\subset U$ the ball with centre at the origin and
radius $r$, and by $\|\cdot\|_\Sigma$ the Euclidean norm with weight
$\Sigma$. 

\begin{proposition}[Well-posedness of the Bayesian inverse problem] \label{prop:inv_wellposed}
Assume that $\mathbb{Q}_0$ is a Gaussian distribution with covariance $\Sigma$,
and that the log-likelihood is $\Phi(p;y) = \|y - G(p)\|_\Sigma^2$. Assume that
the observation operator $G$ satisfies the following conditions.
\begin{enumerate}[i)]
\item For all $\varepsilon$, there exists $M\in \mathbb{R}$ such that for all
$p\in U$, 
\[
\|G(p)\|_\Sigma \leq \exp\left( \varepsilon \|p\|_U^2 + M \right).
\]

\item For all $r>0$, there exists a $K$ such that for all $p_1,p_2\in B(0,r)$,
\[
\|G(p_1) - G(p_2)\|_\Sigma \leq K\|p_1 - p_2\|_U.
\]
\end{enumerate}

Then, the distribution $\mathbb{P}^y(p)$ defined in Theorem \ref{prop:bayes}
exists, and depends continuously on the data $y$.
\end{proposition}
\begin{proof}  See \cite[theorem 4.2]{stuart10}.
\end{proof}

Note that the first condition is a bound on the growth of the solution with
respect to the parameter $p$, whilst the second condition is the Lipschitz
continuity of the observation operator. The first condition is not necessary to
obtain well-posedness if the parameter space $U$ is finite dimensional.

We have characterised the posterior distribution -- the probability distribution
of the parameter $p$ given the data $y$ -- in terms of the data, the model, the
noise distribution and the prior distribution. Using this characterisation of
the posterior, we shall use numerical tools to obtain useful information about
the parameter. Note that in many cases it is crucial to know if the parameter is
well-determined by the data \cite{gutenkunst07}, and often the precise value of
it is not required to draw a conclusion \cite{ashyraliyev08}. For instance, we
can compute credible regions, which are regions of the parameter space with a
given probability. Credible regions are the analogue of confidence intervals in
the Bayesian framework.

The direct or deterministic approach, connected to the optimal control methods (see for instance
\cite{croft15}), is to find the minimum of the log-likelihood
function $\Phi$ given in \eqref{eqn:logl}. This approach is known as the maximum likelihood
method (due to the minus sign in the definition of the log-likelihood, a
minimum of the log-likelihood is a maximum of the likelihood). The
log-likelihood can be interpreted as a distance between the data and the
solution of the mathematical model, so a natural approach is to minimise this
distance. For example, if the noise distribution $\mathbb{Q}_0$ is Gaussian
with mean $0$ and covariance $\Sigma$, we have $\Phi(p;y) = \|y -
G(p)\|_\Sigma^2$, i.e. the log-likelihood is the squared Euclidean distance
with weight $\Sigma$.  

The maximum likelihood method gives a value for the parameter, but does not
provide any information about the uncertainty; the maximum likelihood is only a
fraction of the information encoded in the posterior. Furthermore, a numerical
algorithm could find a local maximum, so global optimisation algorithms may be
necessary \cite{ashyraliyev09}.

In order to use more information from the posterior, we have to approximate
numerically the full posterior distribution. As long as we are able to obtain a
good approximation, we can then extract any information encoded in the
posterior from the approximation. To numerically approximate a probability
distribution we shall generate a finite set of points distributed as samples from
the target distribution.

A useful mathematical structure to approximate a probability distribution is a
Markov chain. A Markov chain is a sequence such that the probability of the
next value, termed the next state, depends only on the present value. Markov
chain Monte Carlo methods (MCMC) are a family of methods that produce a Markov
chain with a given distribution \cite{norris98}. For the problem at hand, the target
distribution is the posterior.

Metropolis--Hastings methods \cite{metropolis53,hastings70} are a special case
of the MCMC methods defined by a proposal kernel
$k(x,y)$, the probability to propose a state $y$ if the present state is
$x$, and an acceptance probability $a(x,y)$, the probability of accepting
the proposed state $y$ if the present state is $x$. To generate a new state
of the chain, one proposes a new state $y$ sampling from the proposal kernel
$k(x,y)$, and then one accepts the proposal with probability $a(x,y)$. If the
proposal $y$ is accepted, it becomes the next state. Otherwise, we stay in the
present state, i.e. the next state is $x$. See \cite{berg08} for a general
introduction to MCMC methods. A simple choice for the proposal kernel is to sample from the prior
distribution, i.e. $k(x,y) = \mu_0(y)$. More elaborate choices are possible
(see for instance \cite{cotter13}). For symmetric proposal kernels, the
acceptance probability is given by $a(x,y) = \min(\exp(\Phi(x)-\Phi(y)),1)$.
Algorithm \ref{alg:classic-mh} describes one step of the Metropolis-Hastings
algorithm. 

\begin{algorithm}
\begin{algorithmic}
\State 1. Draw a proposal  $\bar{x}$ from the proposal kernel
$k(x_k,\bar{x})$.
\State 2. Evaluate $a(x_k,\bar{x}) = \min(\exp(\Phi(x_k)-\Phi(\bar{x})),1)$.
\State 3. Accept $\bar{x}$ with probability $a(x_k,\bar{x})$, i.e
$x_{k+1}=\bar{x}$ with probability $a(x_k,\bar{x})$, and $x_{k+1}=x_k$
otherwise.
\end{algorithmic}
\caption{One step of the standard Metropolis--Hastings algorithm. This will
generate one new sample.}
\label{alg:classic-mh}
\end{algorithm}

MCMC methods are robust but slow. The distribution of the Markov chain
converges, under general conditions, to the target distribution, but long
chains are necessary to obtain good approximations \cite{norris98}. Also, the
methods are inherently sequential. Furthermore, in the parameter identification
framework, evaluation of the acceptance probability typically involves, at
least, one evaluation of the log-likelihood, and in consequence, one evaluation
of the observation operator $G$. This is an expensive operation in terms of
computing time in cases where the model is a system of PDEs because it entails solving the system. Similarly, methods that require quantities like $\nabla \Phi$ may not be
feasible.

To overcome these difficulties we shall use a parallel Metropolis--Hastings
algorithm \cite{calderhead14,tjelmeland04}. The key idea is to generate an
\textit{N-dimensional} Markov chain, such that its distribution is $N$ copies
of the target distribution. This can be done in a way that allows parallel
evaluations of the log-likelihood. 

The parallel Metropolis--Hastings algorithm goes as follows. Let $k(x,y)$ be a
proposal kernel---the probability to propose a state $y$ if the present state
is $x$.  For ease of presentation assume that the proposal kernel is
symmetric, i.e. $k(x,y)=k(y,x)$. Define $r(x,y) = \exp(\Phi(x)-\Phi(y))$. Let
$x_i$ be the present state. 

The algorithm is applied as follows. Given a present state $x_k$, we generate
$N$ new proposals $\{\bar{x}^j\}_{j=1}^N$ from the proposal kernel
$k(x_k,\bar{x}^j)$. Take $\bar{x}^0 = x_k$. Then, in parallel, we evaluate the log-likelihoods
$\Phi(\bar{x}^j)$, $j=1\ldots N$. Note that this will be done by $N$ instances
of the PDE solver running in parallel, but the solver itself need not to be
parallel. With the values of $\Phi(\bar{x}^j)$ at hand, the acceptance
probability of each proposal is computed by finding the stationary distribution
of a Markov chain with $N+1$ states, given by the transition matrix
\[ A(i,j) = \left\{ \begin{aligned} \frac{1}{N}\min(1,r(\bar{x}^i,\bar{x}^j))\quad & \textrm{ if }
j\neq i, \\ 1 - \sum_{j\neq i} A(i,j) \quad & \textrm{ if } j=i. \end{aligned}
\right.  \]
Finally, we sample $N$ times from the stationary distribution to produce $N$
new states (see Algorithm \ref{alg:pmh}). Note this approach is nonintrusive and does not require modifications of the solver for the PDEs.

\begin{algorithm}
\begin{algorithmic}
\State 1. Draw $N$ new points $\{\bar{x}^j\}_{j=1}^N$ from the proposal kernel
$k(x_k,\bar{x}^j)$. Take $\bar{x}^0 = x_k$.
\State 2. Evaluate, in parallel, $\Phi(\bar{x}^j)$, $j=1\ldots N$.
\State 3. Compute the acceptance probabilities of the proposals, given by the
stationary distribution of a Markov chain with transition matrix
\[ A(i,j) = \left\{ \begin{aligned} \frac{1}{N}\min(1,r(\bar{x}^i,\bar{x}^j))\quad & \textrm{ if }
j\neq i, \\ 1 - \sum_{j\neq i} A(i,j) \quad & \textrm{ if } j=i. \end{aligned}
\right.  \]
\State 4. Sample $N$ times from the stationary distribution to produce
$x_{k+1},\ldots,x_{k+N}$.
\end{algorithmic}
\caption{One step of the parallel Metropolis--Hastings algorithm. This will
generate $N$ new samples.}
\label{alg:pmh}
\end{algorithm}

The method described above has been implemented in Python, using the module
Scipy \cite{jones01} and the module \textit{multiprocessing} from the Python Standard
Library. The evaluation of the log-likelihood can be implemented using a different
language and then invoked from the Python code. In this paper, we use Python
also for the evaluation of the log-likelihood; see Section
\ref{sec:application} for details.

Our primary long term motivation is to work with experimental data, for this case one can use the proposed methodology as
follows. Using the data, implement a routine to evaluate the log-likelihood.
Since we are going to explore a range of possible parameters, the evaluation of
the log-likelihood must be robust. Apply the parallel Metropolis--Hastings to
generate several long Markov-Chains, for instance 10 chains of length $10^6$.
Check the convergence of the chains, by checking the convergence of the means,
or using other methods (see for instance \cite{raftery95}). Discard a fraction
at the beginning of each chain, where the effects of the choice of the initial
state were still evident. Finally, recombine the chains in a large set, that approximates
our target distribution. Note that using the parallel Metropolis-Hastings
algorithm, we are able to generate longer chains than using the standard
Metropolis-Hastings algorithm. 

Following this approach we obtain a numerical approximation of the posterior
distribution given by the application of Bayes' theorem to the parameter
identification problem. Note that the method is non-intrusive: the solver for
the model is independent of the parameter identification algorithms. Hence, it
can be applied to a wide range of problems in particular those for which an
existing fast and possibly complex solver for the forward problem is available.
As a proof of concept of our approach, we will consider a well studied
mathematical model of the form of reaction-diffusion systems to perform the
parameter identification. We are motivated by the fact that a precise
definition of the parameter space  in which, for example, patterns can form for
some specific reaction-kinetics (e.g. {\it activator-depleted} type) is known analytically and therefore, we have a clean example to validate the accuracy of our methodology. We present a concrete example in the next section. 
 
 \subsection{HPC computations}
For each of the examples shown in Section \ref{sec:application},  we generate 10 Markov chains for a total of
approximately $10^6$ samples. The mean and the correlation of each chain are
examined, and used to decide the burn-in---the fraction of the chain discarded
due to the influence of the initial value. The burn-in fraction is determined
by checking the convergence in mean for each chain. Finally, the chains are combined
in a big set, that we use to generate the plots.

All the results shown are generated using the local
HPC cluster provided and managed by the University of Sussex. This HPC
cluster consists of 3000 computational units. The models of the
computational units are AMD64, x86\_64 or 64 bit architecture, made up of
a mixture of Intel and AMD nodes varying from 8 cores up to 64 cores per
node. Each unit is associated with 2GB memory space. Most of the
simulations in this paper are executed using 8-48 units.

The wall-clock computation time for one chain ranged from 1 to 4 days, CPU time
ranged from 10 to 35 days for one chain. It must be noted that
there are two levels of parallelisation: first the algorithm is parallelised
and uses the 8 available cores, and then many instances of the algorithm run at
the same time to produce independent chains. We remark that we produce independent
chains in order to test the convergence of the algorithm.

\section{Applications and results}\label{sec:application}
To illustrate the methodology proposed in this study, we consider
infinite (as well as finite) dimensional parameter identification for a well-known mathematical model, a
reaction-diffusion system on evolving and stationary domains. As mentioned earlier, we take the Schnakenberg kinetics \cite{gierer72,prigogine68,schnakenberg79} for illustrative purposes. We use our Bayesian parameter identification approach to find the growth function associated with the reaction-diffusion system as outlined next.

\subsection{Reaction-diffusion system posed on uniform isotropic evolving domains}
Let $\Omega_t  \subset {\mathbb{R}}^m$ $(m=1,2)$   be a simply connected bounded evolving domain for all time $t \in I=[0,t_F]$, $t_F > 0$ and ${\partial \Omega_t}$ be the evolving boundary enclosing $\Omega_t$. Also let $ {\bfa  u} = \left( u \left( {\bfa x} (t),t \right), v \left( {\bfa x} (t),t \right) \right)^T$ be a vector of two chemical concentrations at position ${\bfa x} (t) \in \Omega_t \subset  {\mathbb{R}}^m$. The growth of the domain $\Omega_t$ generates a flow of velocity ${\bfa v}$. For simplicity, let us assume a uniform isotropic growth of the domain defined by
$ {\bfa x} (t) = \rho (t) {\bfa x} (0)$
where, ${\bfa x}(0)  \in \Omega_0$ is the initial domain and $\rho (t)\in
C^1(0,1)$ is the growth function (typically exponential, linear or logistic). We further assume that the flow velocity is identical to the domain mesh velocity
$ {\bfa v} : = \frac{d {\bfa x}}{dt}.$
The evolution equations for reaction-diffusion systems in the absence of
cross-diffusion can be obtained from the application of the law of mass
conservation in an elemental volume using Reynolds transport theorem. Since the
domain evolution is known explicitly, a Lagrangian mapping from an evolving to
a stationary reference initial domain yields the following non-dimensional
reaction-diffusion system with time-dependent coefficients
\cite{crampin02,mackenzie11,madzvamuse07,madzvamuse10,madzvamuse16}
\begin{equation} 
 \begin{cases}
 \begin{cases}
  u_t +\frac{m \dot{\rho} (t)}{\rho (t)} u= \frac{1}{\rho^2 (t)} \nabla^2 u  +  \gamma f(u,v), \\  
  \\
  v_t + \frac{m \dot{\rho} (t)}{\rho (t)} v= \frac{d}{\rho^2 (t)} \nabla^2 v + \gamma g(u,v), 
  \end{cases}
  \quad {\bfa x} \in \Omega_0, \\
  \\
{\bfa n} \cdot \nabla u= {\bfa n} \cdot \nabla v =0,\; {\bfa x}\,\text{on}\,\partial\Omega_0,  \\ \\
u({\bfa x},0)= u_0({\bfa x}),\; \text{and} \; v({\bfa x},0)=v_0({\bfa x}),\; {\bfa x}\,\text{on}\;\Omega_0, 
 \end{cases}
 \label{eqn:schnakenberg}
\end{equation}
where $\nabla^2$ is the Laplace operator on domains and volumes, $d$ is the ratio of the diffusion coefficients and $\dot{\rho}:=\frac{d \rho}{dt}$. Here, {\bfa n} is the unit outward normal to $\Omega_t$. Initial conditions are prescribed through non-negative bounded functions $u_0 ({\bfa x})$ and $w_0 ({\bfa x})$.  In the above, $f(u,v)$ and $g(u,v)$ represent nonlinear reactions and these are given by the  {\it activator-depleted}  kinetics 
 \cite{crampin99,gierer72,prigogine68,schnakenberg79}
\begin{equation}\label{eqn:reactions}
f(u,v)  = a - u + u^2 v, \quad \text{and} \quad 
g(u,v)  = b - u^2 v.
\end{equation}
To proceed, let us fix the parameters $a$, $b$, $d$ and $\gamma$, and use the Bayesian approach to identify the domain growth rate function $\rho (t)$, assuming $\rho(0) = 1$,  $\rho(t)> 0$ for $t>0$ and at the final fixed domain size $\rho(T)$ is known. 

Well posedness results for the system of equations \eqref{eqn:schnakenberg}, as
well as stability results, can
be found in \cite{venkataraman12}. Furthermore, the positivity of solutions is stablished. These results
ensure that the conditions in Proposition \ref{prop:inv_wellposed} are
satisfied, thus we can conclude well posedness of the parameter identification
problem. 

\begin{proposition}[Well-posedness of the parameter identification problem]\label{prop:sch_wellposed}
Let $G$ be the observation operator associated with the reaction-diffusion system
\eqref{eqn:schnakenberg} with reaction-kinetics \eqref{eqn:reactions}. Assume
that $a$, $b$, $d$, $\gamma > 0$, and $\rho\in C^1(0,T)$, $\rho>0$. Then, the conditions in Proposition
\ref{prop:inv_wellposed} are satisfied, and thus the parameter identification
problem is well-posed.  
\end{proposition} 
\begin{proof} The result follows from the analysis in \cite{venkataraman12},
where well-posedness for the problem \eqref{eqn:schnakenberg} is shown.
In particular, the first condition in Proposition \ref{prop:inv_wellposed} is a
consequence of the \it{a priori} estimates for the solution. 
\end{proof}

Although this example is synthetic, it is biologically plausible. First, we generate
the data by adding noise to a numerical solution of the system, but the
parameter identification problem would be analogous in the case of experimental
data (see for instance \cite{vigil92} for an experiment that produces Turing
patterns similar to the patterns that we use as data in Example 2). 

\subsubsection{Example 1: Infinite dimensional parameter identification} 
Without any loss of generality, we restrict our first example to the one-dimensional case ($m=1$) where we fix model parameters with standard values in the literature as \cite{murray13}
\[a =0.1, \; b= 0.9, \; d= 10, \; \gamma = 1000.\]
We want to identify the growth function $\rho (t)$ given synthetic data generated from two different growth profiles defined by the exponential and logistic functions 
\begin{align}\label{eqn:grates}
\rho_{exp} (t)  = \exp(0.001t),\quad \text{and} \quad 
\rho_{log} (t)  = \frac{\exp(0.01t)}{1+\frac{\exp(0.01t)-1}{\exp(0.006)}}.
\end{align}
We fix the final time $T=600$. Note that the size of the domain at time $T$ is the same in both cases.

Initial conditions are taken as small fixed random perturbations of the homogeneous
steady state, $(u_*,v_*)=\left(a+b, \frac{b}{(a+b)^2} \right)$ 
\begin{align}\label{eqn:ic_gturing}
u_0(x)  = 1 + 0.005\sum_{k=1}^9 \sin(k\pi x),\quad \text{and} \quad 
v_0(x)  = 0.9 + 0.005\sum_{k=1}^9 \sin(k\pi x).
\end{align}
We solve the system using the finite difference method both in space and time.
The system of partial differential equations is solved on the mapped initial unit square domain using a finite difference scheme. Similar
solutions can be obtained by employing finite element methods for example or any other appropriate numerical
method. The algebraic linear systems are solved using a conjugate gradient method from
the module \textit{SciPy} \cite{jones01}.

The time-stepping scheme is based on  a modified implicit-explicit (IMEX) time-stepping scheme where we treat the diffusion part and any linear terms implicitly, and use a single Picard iterate to linearise nonlinear terms \cite{ruuth95,madzvamuse06,venkataraman12}. The method was analysed in \cite{lakkis13} for finite element discretisation. It is a first order, semi-implicit backward Euler scheme, given by
\begin{align}
\begin{split}
\frac{u^{n+1} - u^n}{\tau} + \frac{\dot{\rho}^n}{\rho^n} u^{n+1} & = \frac{1}{\rho^{2n}} \Delta_h u^{n+1} + \gamma (a - u^{n+1} + u^n
u^{n+1} v^n), \\
\frac{v^{n+1} - v^n}{\tau} + \frac{\dot{\rho}^n}{\rho^n} v^{n+1} & = \frac{d}{\rho^{2n}} \Delta_h u^{n+1} + \gamma (b - (u^n)^2
v^{n+1}), 
\end{split}
\end{align}
where $\Delta_h$ is the standard 3-point (or 5-point)  stencil finite difference
approximation of the Laplacian operator in 1-D (or 2-D) respectively, with Neumann boundary conditions.
The parameters of the solver are $h=10^{-2}$ and $\tau=10^{-4}$ for all the
computations.

We remark that our aim here is to illustrate the applicability of the Bayesian
approach and the Monte Carlo methods for parameter identification to problems emanating from experimental sciences. More sophisticated solvers can be used and will improve the computational time \cite{venkataraman13}. The synthetic data is generated by solving the reaction-diffusion system \eqref{eqn:schnakenberg} with reaction-kinetics \eqref{eqn:reactions} up to the final time $T=600$, and then construct the synthetic data by perturbing the solution with Gaussian noise with mean zero and
standard deviation equal to 5\% of the range of the solution  and this data is illustrated in Figure
\ref{fig:data_gturing}.

\begin{figure}[!h] 
\centering
\includegraphics[width=.90\textwidth]{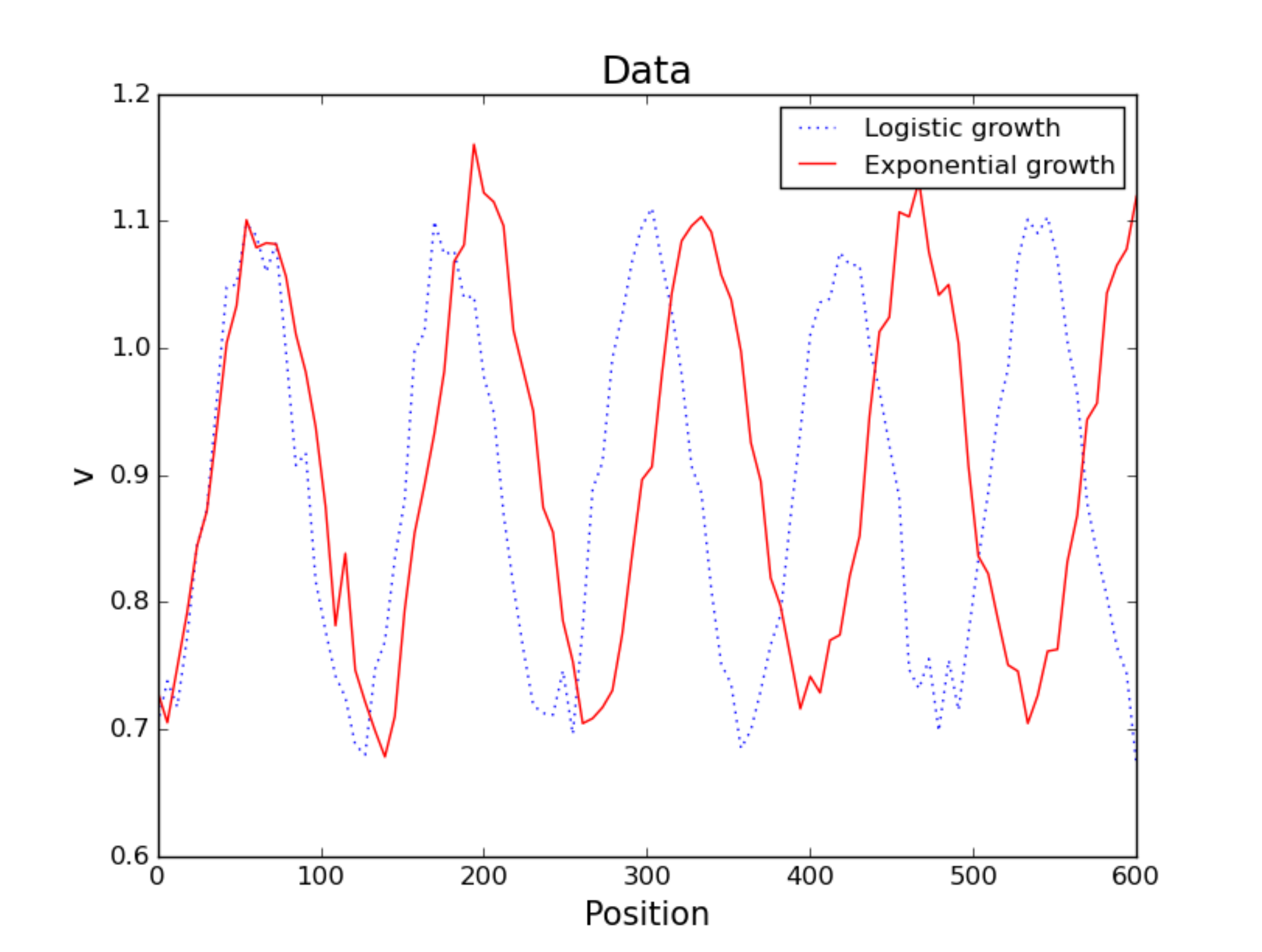}
\caption{Synthetic data for the reaction-diffusion system \eqref{eqn:schnakenberg} with reaction-kinetics \eqref{eqn:reactions} on a one-dimensional growing domain with exponential and logistic growth. The figure depicts only the $v$-component of the solution, the $u$ component is 180 degrees out of phase. (Colour version online).}
\label{fig:data_gturing}
\end{figure}

In order to identify a time-dependent parameter $\rho (t)$, we approximate it on a finite
dimensional space by using a polynomial of degree four, with only three degrees of freedom. The coefficients
of order zero and four are fixed in order to satisfy the conditions for $\rho (t)$
at times $t=0$ and $t=T$, respectively. The prior for the coefficients of the polynomial approximation of $\rho (t)$ is
selected by looking at the region of time-position space covered by $95\%$ of the samples of the
prior as shown in Figure \ref{fig:growth} (large shaded region, yellow in color
version). 
\begin{figure}[!h] 
\centering
\includegraphics[width=.90\textwidth]{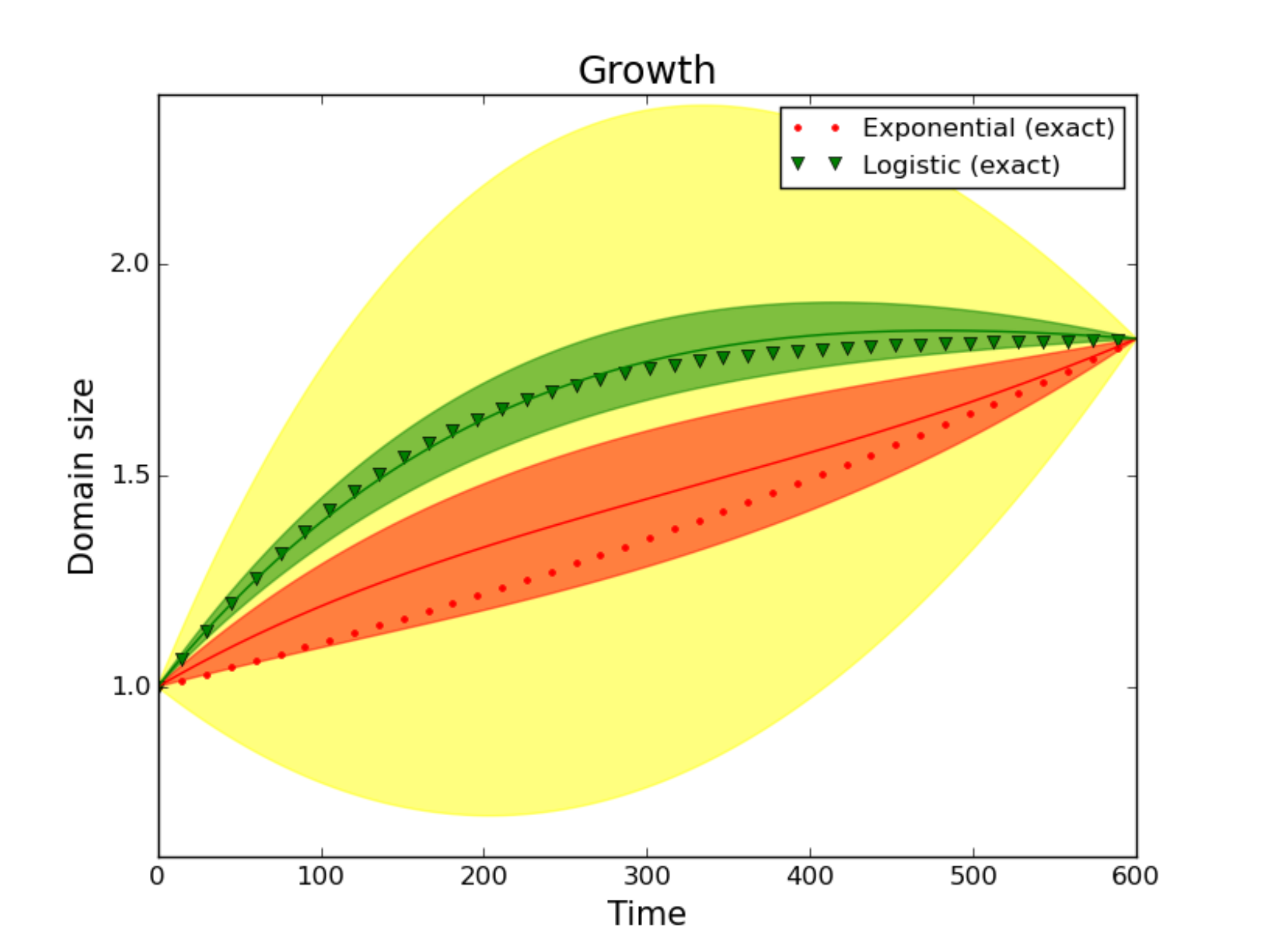}
\caption{Regions where $95\%$ of the samples of the prior (light colour, larger
region) and the posteriors for exponential growth (darker colour, bottom region)
and logistic growth (darker colour, upper region). The exact growth used to
generate the data is traced with triangles (logistic) and circles
(exponential). The solid lines are the growth rates computed using the mean of
each posterior distribution 
 (Colour version online).
}
\label{fig:growth}
\end{figure}

The same prior is used for both the logistic and
the exponential growth data sets. In Figure \ref{fig:growth} we depict the regions where 95\% of the samples from
the posteriors lie, and also the region where 95\% of the samples of the prior
lie. Observe that we could also plot the credible regions for the coefficients of
the finite dimensional approximation of $\rho$, but it is more difficult to
visualise the result from it. 

For the final time $T=600$, we find that the region corresponding to the
parameters identified from the exponential growth rate and the region
corresponding to the logistic growth rate
are completely separated: we can distinguish the type of growth from the
solutions at a final time $T$.

\subsubsection{Example 2: Finite dimensional parameter identification}
Next we demonstrate the applicability of our approach to identifying credible regions for parameters which are constant and not space nor time-dependent and this is a prevalent problem in parameter identification. We will again use the reaction-diffusion system \eqref{eqn:schnakenberg} with reaction-kinetics \eqref{eqn:reactions}  in the absence of domain growth, i.e. $\rho (t)=1$ for all time. Our model system is therefore posed on stationary domains, for the purpose of demonstration, we assume a unit-square domain. We seek to identify $a$, $b$, $d$ and $\gamma$.  For ease of exposition, we will seek  parameters in a pair-wise fashion. 
\begin{table}
\centering
\begin{tabular}{c c}
\hline
Parameter & Value\\
\hline
$a$ & 0.126779 \\
$b$ & 0.792366 \\
$d$ & 10 \\
$\gamma$ & 1000 \\
\hline
\end{tabular}
\caption{Exact values of the parameters, used to generate the noisy data shown
in Figure \ref{fig:data}.}
\label{tab:param}
\end{table}

Our ``experimental data'' are measurements of the steady state of the system. We fix the
initial conditions as a given perturbation of the spatially homogeneous steady
state given by $(a+b,\frac{b}{(a+b)^2})$. For the values of the parameters given
in Table \ref{tab:param}, the initial conditions are given by
\begin{align} \label{eqn:ic_numerical}
\begin{split} u_0(x,y) & = 0.919145 + 0.0016 \cos(2\pi(x+y)) + 0.01\sum_{j=1}^8
\cos(2\pi j x), \\ v_0(x,y)& = 0.937903 + 0.0016 \cos(2\pi(x+y)) + 0.01\sum_{j=1}^8
\cos(2\pi j x). \end{split}
\end{align}
We let the system evolve until the $L_2$ norm of the discrete time-derivative is smaller than a certain threshold. For the numerical experiments presented here, the threshold is $10^{-5}$.  At that point, we assume
that a spatially inhomogeneous steady state has been reached. We record the
final $T$ time and save the
solution. To confirm that the solution is indeed stationary, we keep solving the
system until time $t=2T$ and check that the difference between the solutions at
time $T$ and $2T$ is below the threshold. To generate our synthetic measurement, we add Gaussian noise to the
solution, as illustrated in Figure \ref{fig:data}.
\begin{figure}[!h] 
\centering
(a) \includegraphics[width=.6\textwidth]{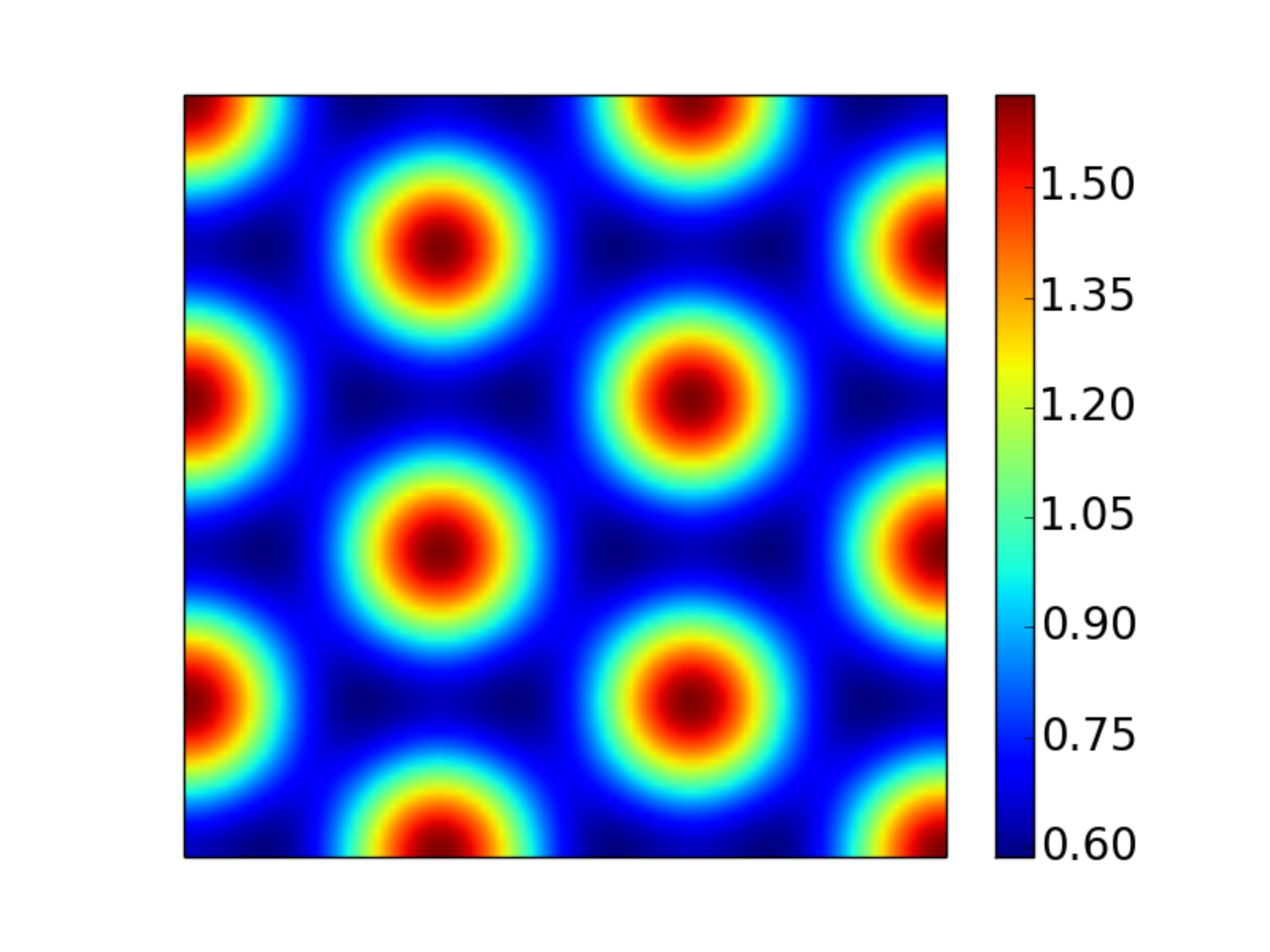}
(b)\includegraphics[width=.6\textwidth]{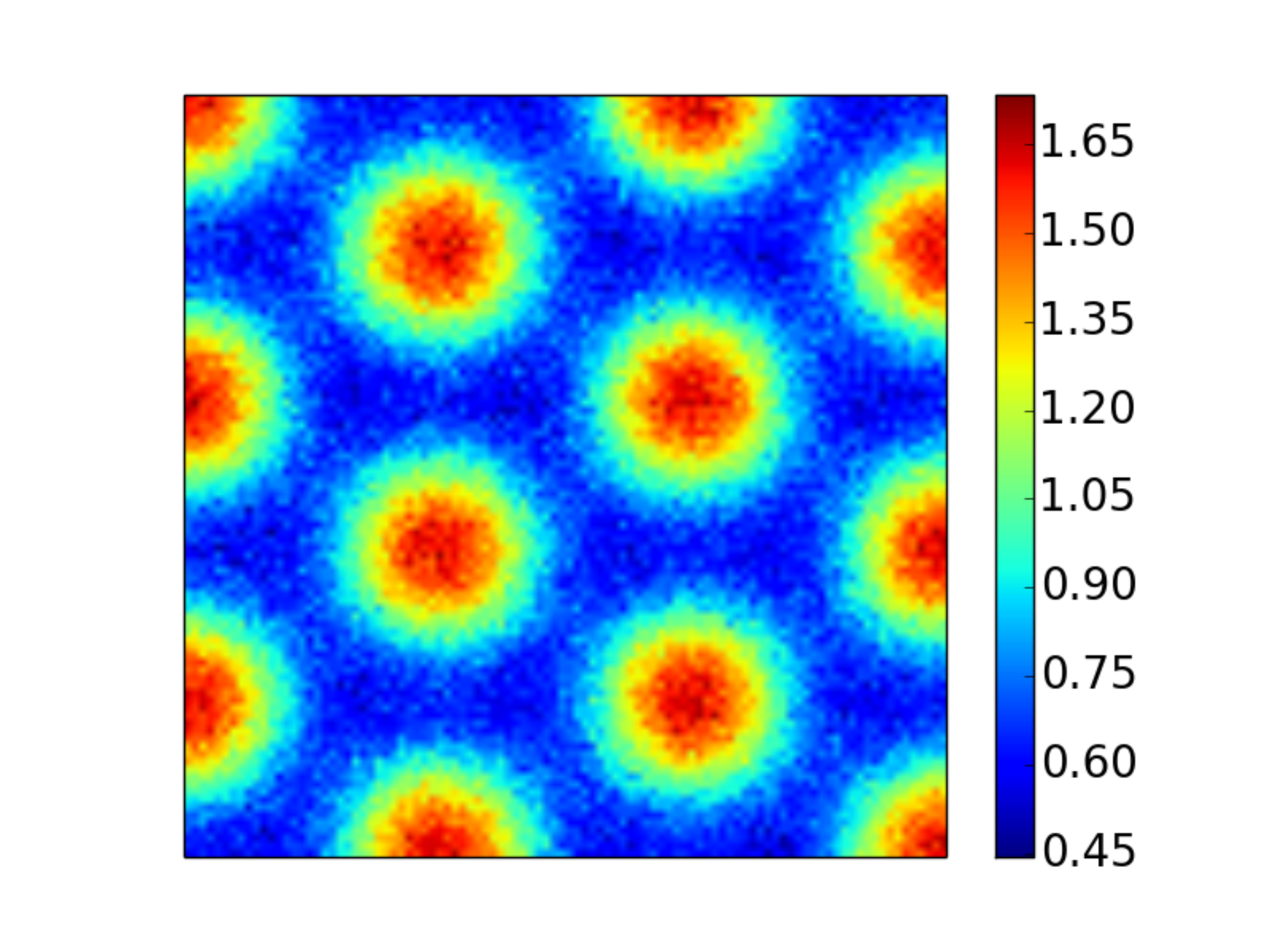}
(c)\includegraphics[width=.6\textwidth]{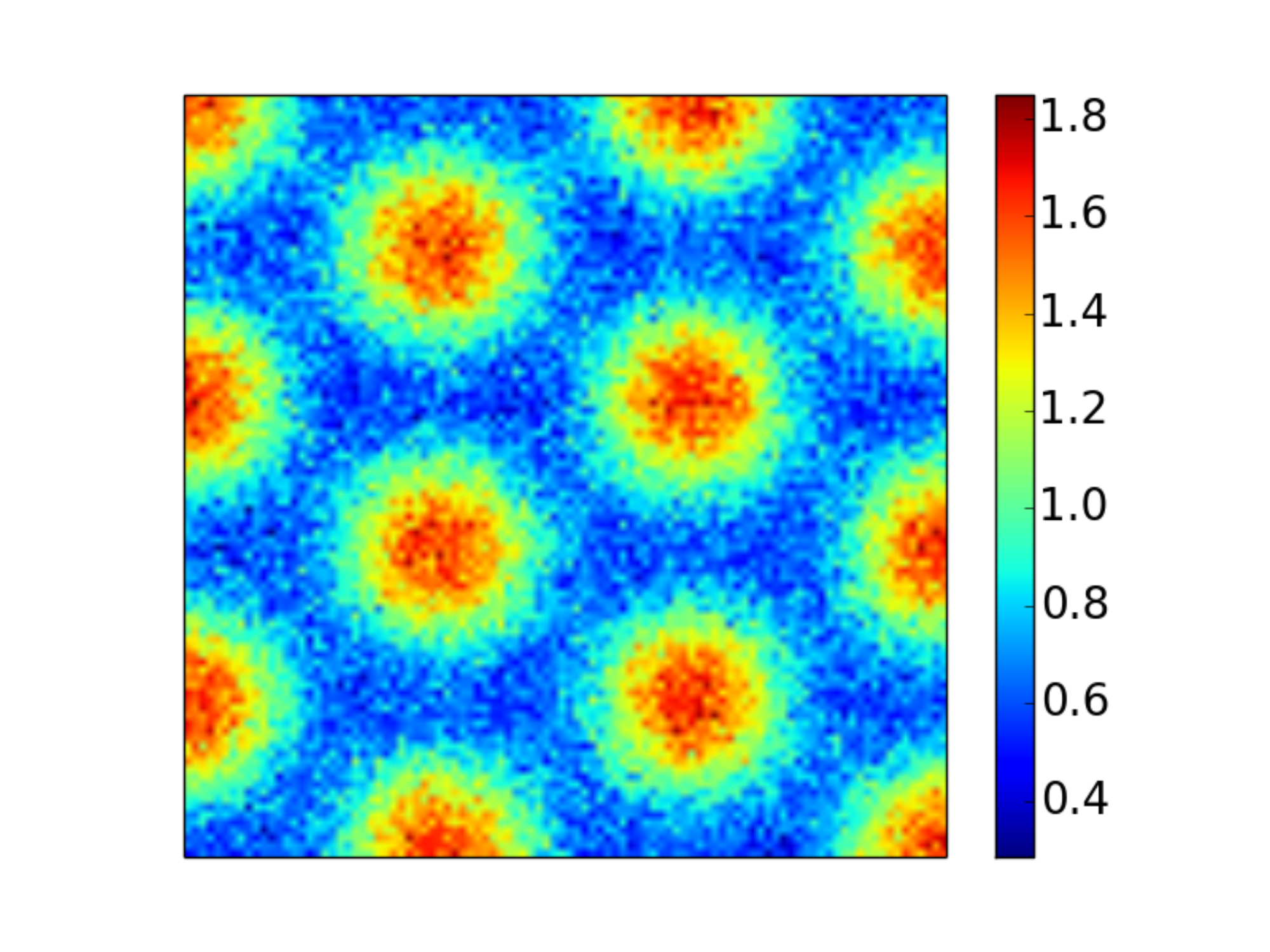}
\caption{The $u$-component of the solution to the Schnakenberg system \cite{gierer72,prigogine68,schnakenberg79}, with added
Gaussian noise with mean zero and standard deviation 5\% (b) and 10\% (c)
of the range of the solution. Solutions of the $v$-component are 180°
out-of-phase with those of $u$ and as such their plots are omitted (Colour version online).} 
\label{fig:data}
\end{figure}

This synthetic experiment is a situation similar to what one will face in an
actual experiment, although some assumptions, in particular the fixed known
initial conditions, are not realistic. A detailed study of the dependence of
the solution with respect to initial conditions will be necessary to drop this
assumption. Alternatively, the initial conditions could be included as a
parameter to identify.

\subsubsection*{Case 1: Credible regions for $a$ and $b$ with little knowledge}
For our first example, we assume that the values of the parameters
$\gamma$ and $d$ are known, and that we would like to find the values for the parameters
$a$ and $b$. In a first approach, we assume very little knowledge about $a$ and $b$:
only their order of magnitude is assumed to be known. This knowledge is
modelled by a uniform prior distribution on the region given by $[0.1,10]^2$.
The data for this example has Gaussian noise with standard deviation 5\% of the
range of the solution. We can see in Figure \ref{fig:uni5} a 95\% probability
region for the posterior distribution. Observe how this region is concentrated
around the exact value, in contrast to our original knowledge of a uniform
distribution on the region $[0.1,10]^2$. More precisely, the credible region
is contained within the range $[0.126,0.128]\times[0.789,0.796]$. The length of the credible
region in the $b$-direction is approximately 3.5 times larger than in the
$a$-direction, although the size relative to the magnitude of the parameters is
smaller in the $b$-direction.

\begin{figure}[!h] 
\centering
\includegraphics[width=.90\textwidth]{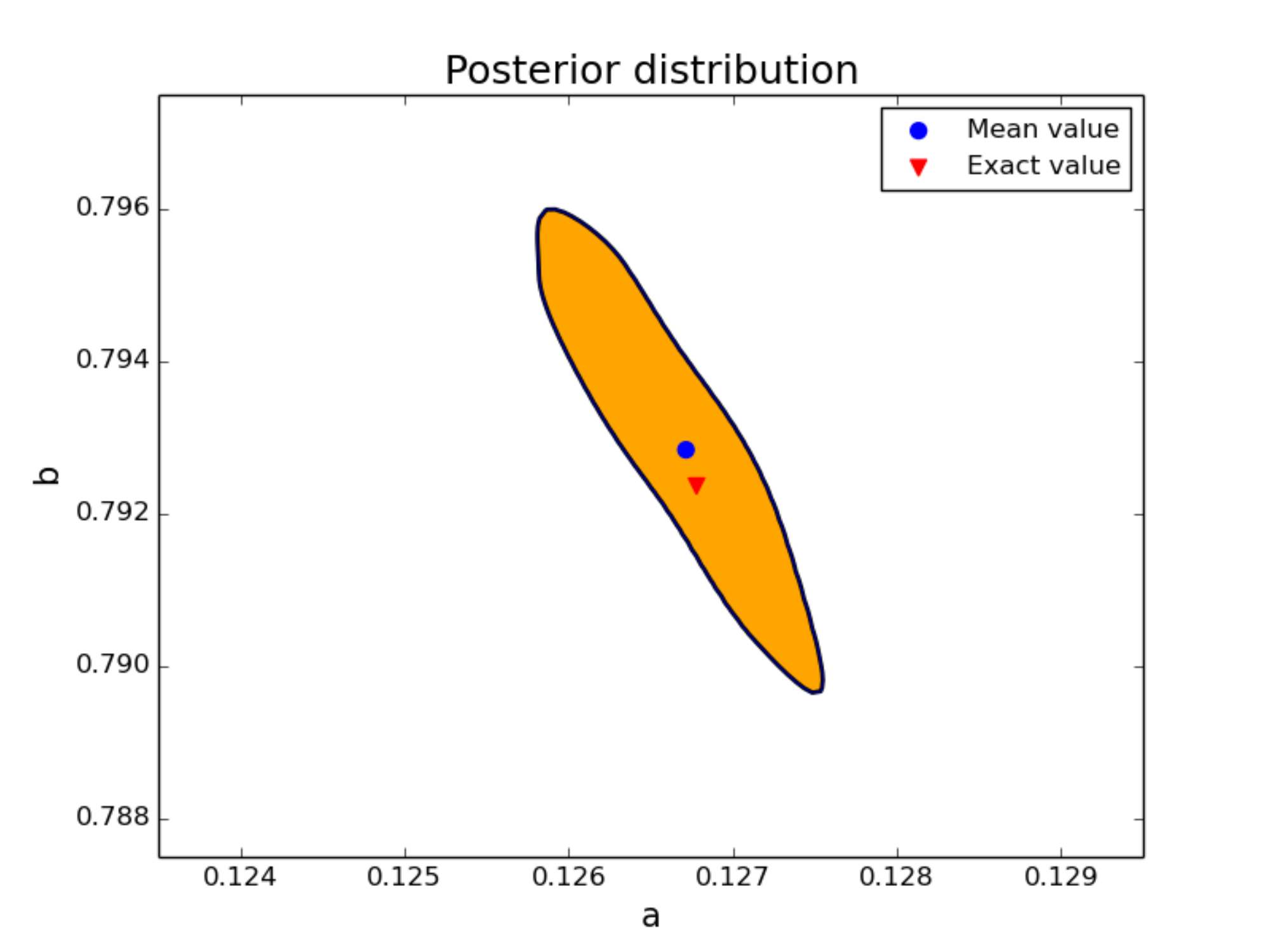}
\caption{95\% credible region for the posterior distribution for the
parameters $a$ and $b$, using a
uniform prior on the region $[0.1,10]^2$. Note that for ease of visualisation, scales for $a$ and $b$ are
different (Colour version online).}
\label{fig:uni5}
\end{figure}

For the Schnakenberg reaction kinetics \cite{gierer72,prigogine68,schnakenberg79}, it is possible to compute the region that contains
the parameters that can lead to non-homogeneous steady states, the Turing
space (see next example for details). We can see that our original prior covered an area much much larger than the
Turing space (almost 100 times larger), but the posterior is concentrated in a small region completely
contained within it (see Figure \ref{fig:turing}).

\begin{figure}[!h] 
\centering
\includegraphics[width=.90\textwidth]{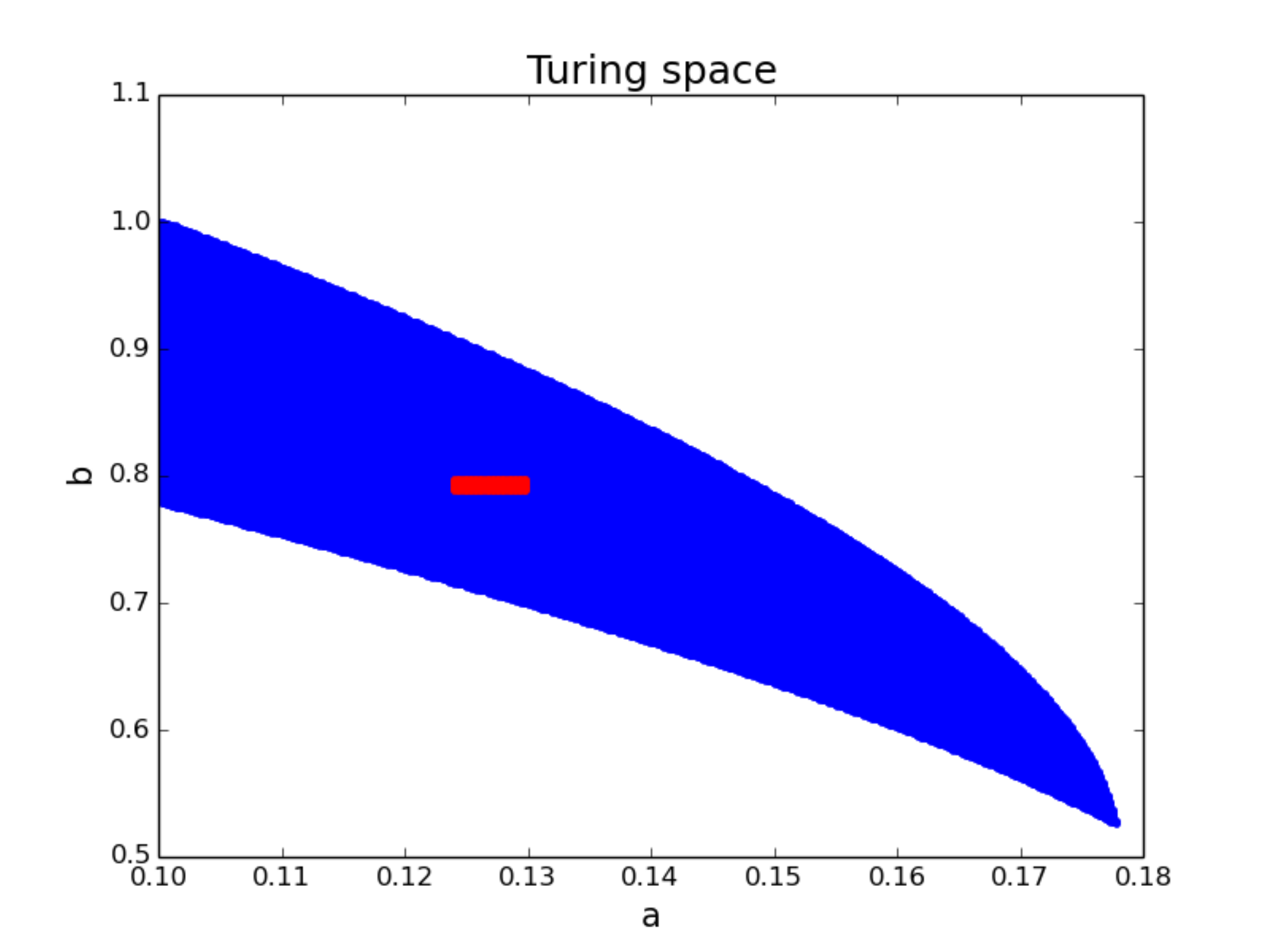}
\caption{Darker region (blue in the online version), the Turing space for the parameters $a$ and $b$ of the Schnakenberg
reaction kinetics \cite{gierer72,prigogine68,schnakenberg79}. Lighter region (red in the online version), the region plotted in Figure \ref{fig:uni5},
depicting where the credible region is contained (Colour version online).}
\label{fig:turing}
\end{figure}

\subsubsection*{Case 2: Credible regions for $a$ and $b$ using the Turing parameter space}
Unlike the previous example, where we assume little knowledge of the prior, here we exploit the well-known theory for reaction-diffusion systems on stationary domains and use a much-more informed prior based on analytical theory of the reaction-diffusion system. On stationary domains, diffusion-driven instability theory requires reaction-diffusion systems to be of the form of {\it long-range inhibition, short-range activation} for patterning  to occur (i.e. $d>1$). More precisely, a necessary condition for Turing pattern formation is that the parameters belong to a well-defined parameter space \cite{murray13} described by the inequalities (for the case of reaction-kinetics \eqref{eqn:reactions}) 
\begin{align}\label{eqn:turing}
\begin{split}
 f_u+g_v  & = \frac{b-a}{b+a} - (a+b)^2 < 0, \\
 f_u g_v - f_v g_u & = (a+b)^2 > 0, \\
 d f_u + g_v & = d\left( \frac{b-a}{b+a} \right) - (a+b)^2 > 0, \\
 (d f_u + g_v)^2 & - 4d(f_u g_v - f_v g_u) \\
&  = \left(d \left( \frac{b-a}{b+a} \right)  -
(a+b)^2 \right)^2 - 4d(a+b)^2 > 0,
\end{split}
\end{align} 
where $f_u$, $f_v$, $g_u$ and $g_v$ denote the partial
derivatives of $f$ and $g$ with respect to $u$ and $v$, evaluated at the spatially
homogeneous steady state $(a+b,\frac{b}{(a+b)^2})$. In Figure
\ref{fig:turing} we plot the parameter space obtained with Schnakenberg kinetics.

In this second example, we use data with added Gaussian noise with standard
deviation 10\% of the solution range. For the prior, we now use a uniform
distribution on the Turing space of the system. In Figure \ref{fig:tur10} we
can now see that despite the increased noise, the improved prior reduced the size
of the 95\% probability region dramatically.

\begin{figure}[!h] 
\centering
\includegraphics[width=.90\textwidth]{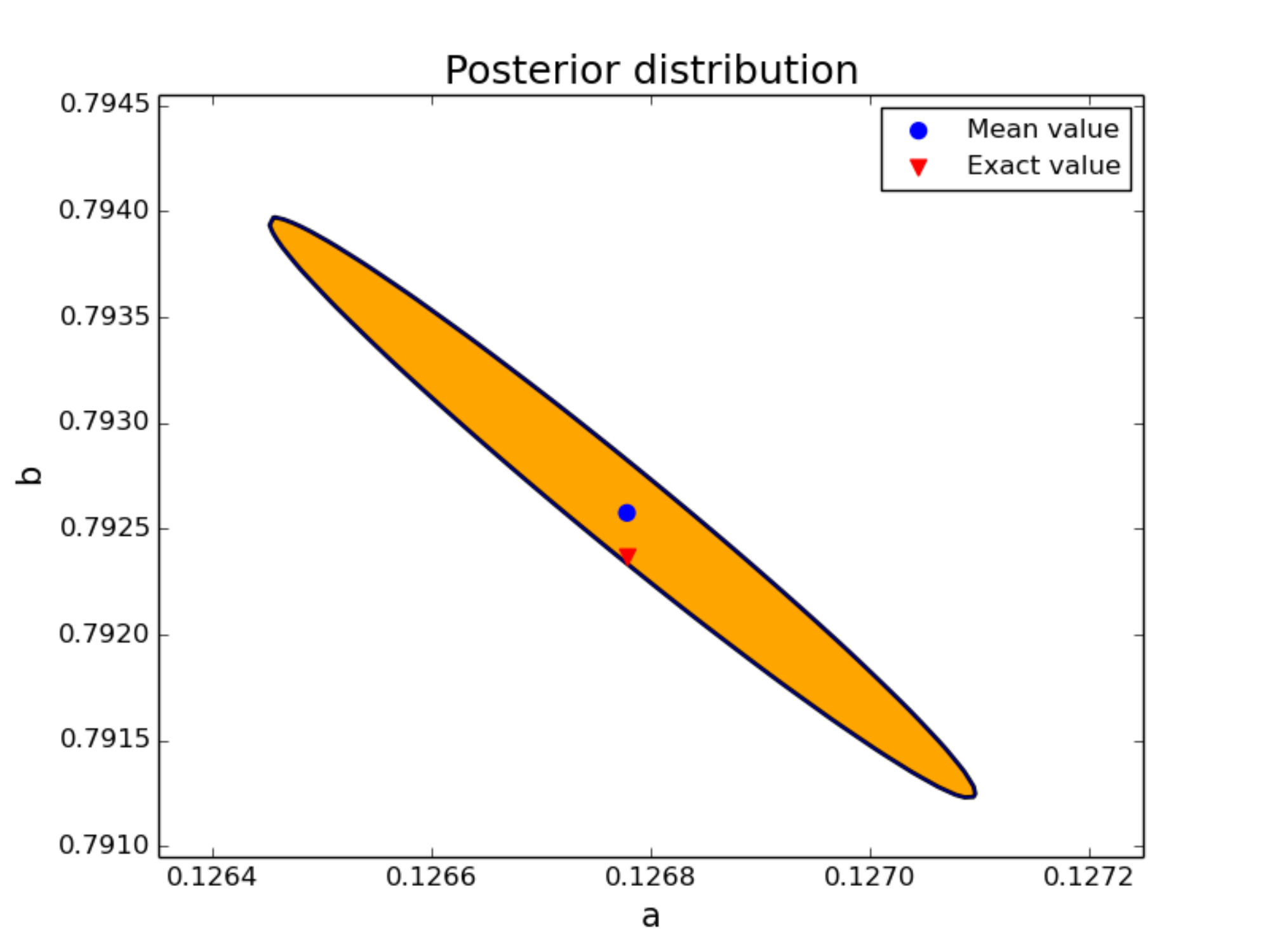}
\caption{95\% credible region for the posterior distribution for the
parameters $a$ and $b$, using a
uniform prior on the Turing space (Colour version online).}
\label{fig:tur10}
\end{figure}

\subsubsection*{Case 3: Credible regions for $d$ and $\gamma$}
In a third example, we assume that $a$ and $b$ are known, and we would like
to find $\gamma$ and $d$. To illustrate the use of different types of priors, here we
assume a log-normal prior that ensures positivity of $\gamma$ and $d$, which in
the case of $d$ is necessary to ensure a well-posed problem. We use the data
with 5\% noise, and the prior distribution of a log-normal with mean (5, 500)
and standard distribution 0.95. A 95\% probability region
of the posterior is depicted in Figure \ref{fig:dg5}. Note the use of a
log-normal or similar prior is essential here to ensure positivity of the
diffusion coefficient which is required for the well posedness of the forward
problem. The prior distribution covers a range of one order of magnitude for
each parameters, and the posterior distribution suggests relative errors of
order $10^{-3}$.

\begin{figure}[!h] 
\centering
\includegraphics[width=.90\textwidth]{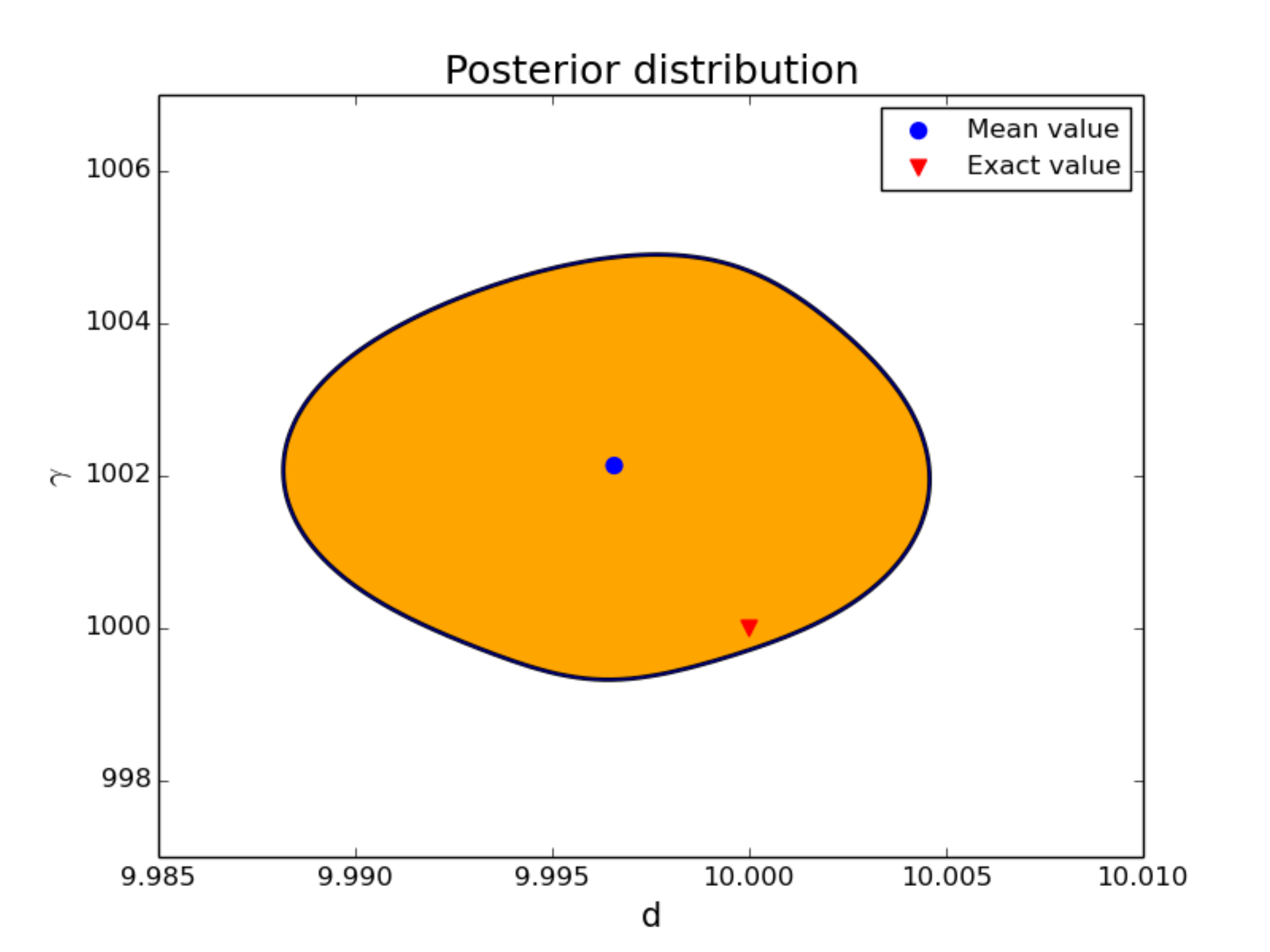}
\caption{95\% credible region for the posterior distribution for the
parameters $d$ and $\gamma$, a log-normal prior with mean (5,500) and standard
deviation 0.95 (Colour version online).}
\label{fig:dg5}
\end{figure}

\subsubsection*{Case 4: Credible regions for $a$, $b$, $d$ and $\gamma$}
Finally, we identify all four parameters $a$, $b$, $d$ and $\gamma$ from the set
of data with $10\%$ noise. We use the priors from {\bf Case 1} for $a$ and $b$,
and from {\bf Case 3} for $d$ and $\gamma$. We do not apply any further
restrictions.

In Figure \ref{fig:all10} we depict the  credible regions for the projection
of the parameters to four different coordinate planes. The noise for this
experiment is higher than in {\bf Case 1} and {\bf 3}. Also note that compared to
the previous experiments we assume here less knowledge {\it a priori} about the
parameters. It must be noted that assuming that a parameter is known is equivalent in the
Bayesian formulation to use a Dirac delta for the prior on the parameter; in
comparison, the prior distributions for {\bf Case 4} are more spread, therefore
they represent less knowledge. Since the level of noise was higher, and the prior knowledge lower, the credible
regions that we obtain are larger.

\begin{figure}[!h] 
\centering
\includegraphics[width=.90\textwidth]{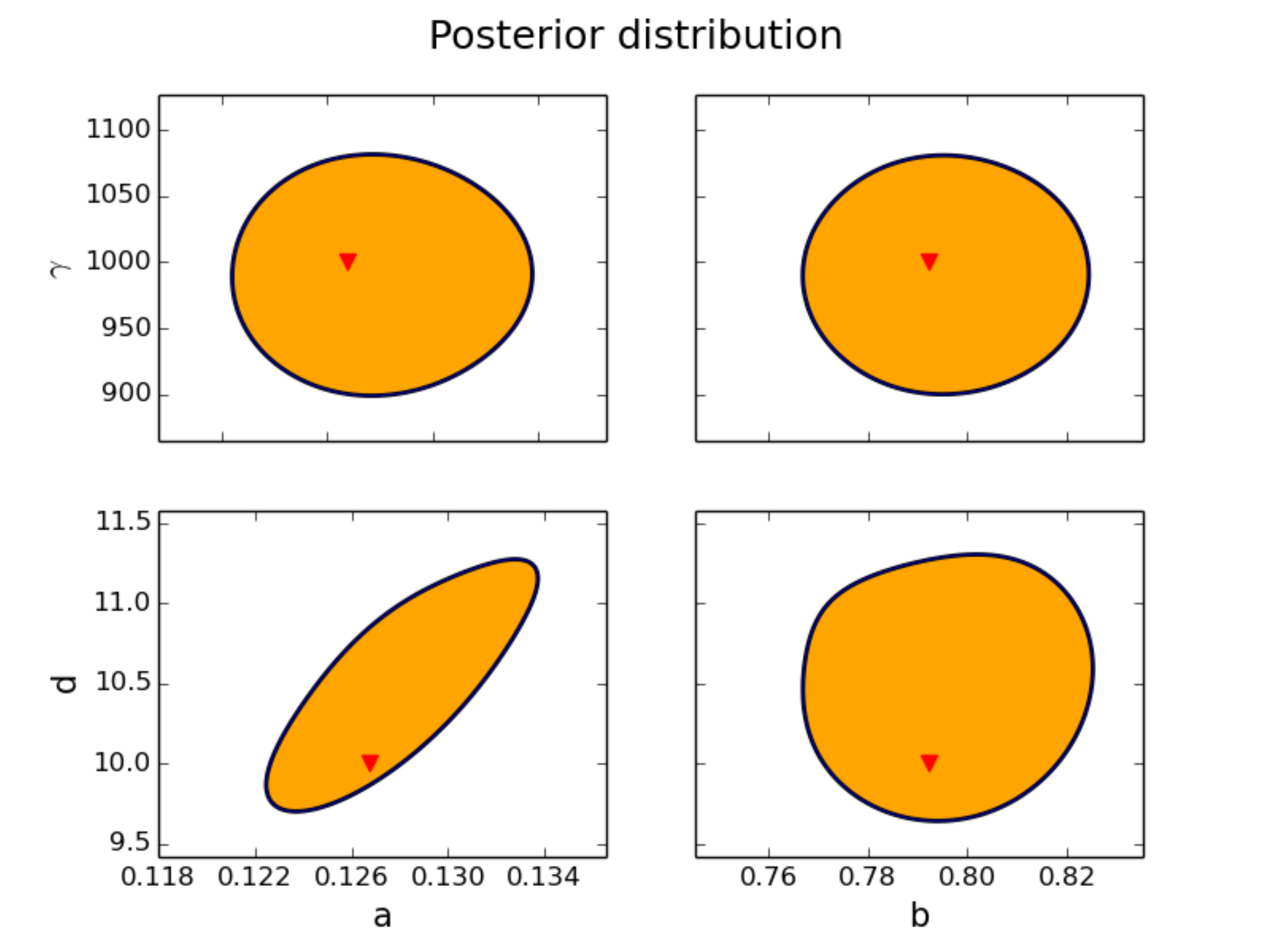}
\caption{95\% credible region for the posterior distribution for the
parameters $a$, $b$, $d$ and $\gamma$. See {\bf Case 1} and {\bf 3} for a
description of the priors. The data has a noise of $10\%$ of the range of the
solution. Each subplot corresponds to the projection of the credible region
onto a coordinate plane for two of the parameters, given by the row and the
column of the subplot. The exact value of the parameters is marked with a triangle.(Colour version online).}
\label{fig:all10}
\end{figure}

\section{Conclusion}
We have presented a methodology for parameter identification based on Bayesian
techniques. The Bayesian framework offers a mathematically rigorous approach
that includes the prior knowledge about the parameters (including functions). Furthermore,
well-posedness results for the problem are available.

Although exploring a whole probability distribution can be computationally
expensive, very useful information about the uncertainty or correlation of the
parameters can be inferred from it. The use of a parallel Metropolis-Hastings
algorithm makes the computations feasible using HPC facilities.

To illustrate the applicability of this technique, we studied the parameter
identification problem for reaction-diffusion systems using Turing patterns as
data for both scenarios: infinite and finite dimensional parameter identification cases. 
We provided a rigorous proof for the well-posedness of the parameter identification problem in the
Bayesian sense, and we performed several numerical simulations to find credible
regions for the parameters.

The mathematical machinery behind the Bayesian framework can be more involved
than in the optimal control approach. For the latter, the general idea is
intuitive: find the parameter that minimises the distance between the data and
the solution. On the other hand, the Bayesian approach requires to understand
the knowledge about a parameter as a probability distribution, to formulate the
problems in terms of conditional probabilities, and then to use Bayes' theorem
to characterise the solution.

An advantage of the Bayesian approach is that it allows us to understand all the
elements of the mathematical formulation of the parameter identification
problem in terms of aspects of the real-world experiment. Whilst in the optimal
control approach one chooses the type of distance, or the norm, and the
regularisation, according to technical limitations of the mathematical methods
(see for instance \cite{blazakis15,portet15,yang15}). In the Bayesian approach these choices are given by the distribution of the
experimental error and the prior knowledge about the parameter. Therefore,
following the Bayesian methodology we can incorporate in the parameter
identification problem more detailed characteristics of the noise and the prior
knowledge, leading to estimations of the parameters more consistent with the
experimental framework. Note that in terms of the optimal control formulation
of the parameter identification problem, a different noise means a different
distance to measure the fit between the data and the model. In general,
different distances will be minimised by different values of the parameters.

\section{Outlook}
In this work, we have treated the synthetic data parameter identification
problem using only information that would be available if the data came from an
actual experiment. We are now in the process of applying this techniques to
problems in experimental sciences and beyond where experimental data is available. The
information provided by the posterior distribution allows us to solve problems
like model selection incorporating all the knowledge from the experiments.

\section{Data accessibility}
All the computational data output is included in the present manuscript.

\section{Competing interests}
We have no competing interests.

\section{Authors contributions}
ECF carried out the bulk of the analysis, algorithm development and
implementation as well as drafting the manuscript as part of his PhD thesis. CV
and AM conceived the study, coordinated, supervised the research study and
contributed equally to the writing up of the manuscript. All authors gave final
approval for publication.

\section{Funding}
The authors (ECF, CV, AM) acknowledge support from the Leverhulme Trust
Research Project Grant (RPG-2014-149).  This work (CV, AM) was partially
supported by the EPSRC grant (EP/J016780/1). The authors (ECF, CV, AM) 
would like to thank the Isaac Newton Institute for Mathematical Sciences 
for its hospitality during the programme [Coupling Geometric PDEs with Physics for Cell Morphology, 
Motility and Pattern Formation] supported by EPSRC Grant Number EP/K032208/1.  
This work (AM) has received funding from the European Union Horizon
2020 research and innovation programme under the Marie Sklodowska-Curie grant
agreement (No 642866). AM was partially supported by a grant from the Simons
Foundation.

\section{Acknowledgements}
The authors acknowledge the support from the University of Sussex ITS for
computational purposes. 

\clearpage
\newpage
\bibliography{library_bibtex}{}
\bibliographystyle{abbrv}

\end{document}